\documentclass[runningheads]{llncs}

\usepackage[utf8]{inputenc}
\usepackage{amsmath}
\usepackage{amsfonts}
\usepackage{amssymb}
\usepackage{authblk}
\usepackage{verbatim}
\usepackage{thmtools}
\usepackage{thm-restate}
\usepackage{hyperref}

\usepackage[capitalize, noabbrev]{cleveref}
\usepackage[autostyle=true]{csquotes}
\usepackage{graphicx}
\usepackage{color}
\usepackage{tabularx}

\usepackage{tikz}
\usetikzlibrary{calc}
\definecolor{lcolor}{RGB}{0,84,159}
\definecolor{fcolor}{RGB}{161,16,53}

\usepackage{orcidlink}

\newcommand{\R}{\mathbb{R}}
\newcommand{\Q}{\mathbb{Q}}
\newcommand{\Z}{\mathbb{Z}}
\newcommand{\N}{\mathbb{N}}
\newcommand{\Rgeq}{\R_{\geq 0}}
\newcommand{\Qgeq}{\Q_{\geq 0}}
\newcommand{\Zgeq}{\Z_{\geq 0}}
\newcommand{\set}[1]{\{ #1 \}}
\newcommand{\fromto}[2]{\set{#1, \ldots, #2}}
\newcommand{\powerset}[1]{2^{#1}}

\DeclareMathOperator*{\argmax}{arg\,max}
\DeclareMathOperator*{\argmin}{arg\,min}
\DeclareMathOperator{\poly}{poly}
\DeclareMathOperator{\conv}{conv}
\DeclareMathOperator{\cone}{cone}

\newcommand{\U}{\mathcal{U}}
\newcommand{\I}{\mathcal{I}}
\newcommand{\F}{\mathcal{F}}
\newcommand{\sol}{\mathcal{S}}

\newcommand{\PTIME}{\textsf{P}}
\newcommand{\NP}{\textsf{NP}}
\newcommand{\SSPNP}{\textsf{SSP-NP}}
\newcommand{\SSPNPc}{\textsf{SSP-NPc}}
\newcommand{\LOP}{\textsf{LOP}}
\newcommand{\leqSSP}{\leq_\text{SSP}}

\newcommand{\sat}{\textsc{Sat}}

\begin{document}

\title{The Complexity of Stackelberg Pricing Games}

\author{
  Christoph Grüne\inst{1}\orcidlink{0000-0002-7789-8870} \and
  Dorothee Henke\inst{2}\orcidlink{0000-0001-9190-642X} \and
  Eva Rotenberg\inst{3}\orcidlink{0000-0001-5853-7909} \and
  Lasse Wulf\inst{3}\orcidlink{0000-0001-7139-4092}
}

\authorrunning{C.\ Grüne et al.}

\institute{
  RWTH Aachen, Germany. \email{gruene@algo.rwth-aachen.de} \and
  University of Passau, Germany. \email{dorothee.henke@uni-passau.de} \and
  IT University of Copenhagen, Denmark. \email{erot@itu.dk} and \email{lasw@itu.dk}
}

\maketitle

\begin{abstract}
We consider Stackelberg pricing games, which are also known as bilevel pricing problems, or combinatorial price-setting problems.
This family of problems consists of games between two players: the leader and the follower.
There is a market that is partitioned into two parts: the part of the leader and the part of the leader's competitors.
The leader controls one part of the market and can freely set the prices for products.
By contrast, the prices of the competitors' products are fixed and known in advance. 
The follower, then, needs to solve a combinatorial optimization problem in order to satisfy their own demands, while comparing the leader's offers to the offers of the competitors.
Therefore, the leader has to hit the intricate balance of making an attractive offer to the follower, while at the same time ensuring that their own profit is maximized.

Pferschy, Nicosia, Pacifici, and Schauer considered the Stackelberg pricing game where the follower solves a knapsack problem.
They raised the question whether this problem is complete for the second level of the polynomial hierarchy, i.e., $\Sigma^p_2$-complete. 
The same conjecture was also made by Böhnlein, Schaudt, and Schauer.
In this paper, we positively settle this conjecture.
Moreover, we show that this result holds actually in a much broader context:
The Stackelberg pricing game is $\Sigma^p_2$-complete for over 50 $\NP$-complete problems, including most classics such as TSP, vertex cover, clique, subset sum, etc.
This result falls in line of recent meta-theorems about higher complexity in the polynomial hierarchy by Grüne and Wulf. 

\end{abstract}

\section{Introduction}
\label{sec:introduction}

\emph{Stackelberg pricing games} are economic games between two players, usually called the leader and the follower.
In these games, we assume that the leader owns one part of the relevant market and may set prices of products to their liking.
The rest of the market belongs to competitors of the leader and has fixed prices, which we assume to be known in advance.
After the leader has set the prices, the follower solves a combinatorial optimization problem, freely combining products of the leader and the competitors. 
The leader therefore has to balance both players' interests carefully.
On the one hand, if the leader sets the prices too low, the leader will not make a large profit. 
On the other hand, if they set the prices too high, the follower will simply switch to the competitors, and the leader will not receive any profit.
Stackelberg pricing games were introduced by Labb{\'e}, Marcotte and Savard, in order to model optimal highway road tolls \cite{labbe1998bilevel}. 
In this context, the follower solves a shortest-path problem. 
Early works have studied $\NP$-hardness, approximation algorithms,  hardness of approximation, or multiple followers for the highway toll pricing problem \cite{briest2012stackelberg,bui2024asymmetry,bui2022catalog,joret2011stackelberg,DBLP:journals/networks/RochSM05}.
Later works have extended this to many other settings where the follower solves different combinatorial optimization problems.
This includes the case of minimum spanning tree (Cardinal et al. \cite{cardinal2011stackelberg}, Labbé et al. \cite{labbe2021computational}),  
knapsack (Briest et al. \cite{briest2012b}, Pferschy et al. \cite{pferschy2021stackelberg}, Bui et al. \cite{bui2025solving}),
set cover (Briest et al. \cite{briest2012b}, Bui et al. \cite{bui2025solving}),
independent set (Böhnlein et al. \cite{bohnlein2023stackelberg}, Bui et al. \cite{bui2025solving}), or the
bipartite vertex cover problem (Briest at al. \cite{briest2012stackelberg}).
A general overview of network pricing problems can be found in the articles of van Hoesel \cite{hoesel2008overview} and of Labbé and Violin \cite{DBLP:journals/anor/LabbeV16}.

\paragraph{Example.} As one possible example of a Stackelberg pricing game, let us mention the
\emph{Stackelberg knapsack game with profit selection}
as defined by Pferschy, Nicosia, Pacifici, and Schauer \cite{pferschy2021stackelberg}. We are given a set $\U$ of items and a partition $\U = \U_L \cup \U_F$ into leader's and follower's items (with $\U_L \cap \U_F = \emptyset$). 
Just like in the standard knapsack problem, each item $e \in \U$ has a weight $w(e)$. 
Together with some given weight threshold $W \geq 0$, this determines the set of \emph{feasible knapsack solutions} by $\F = \set{X \subseteq \U_L \cup \U_F : w(X) \leq W}$. 
Each item has a profit, specified separately for the leader by $p_L : \U_L \to \Zgeq$ and for the follower by $p_F : \U_F \to \Zgeq$. 
The idea is now that, for all items $e \in \U_L$, the leader receives their profit $p_L(e)$ only if the follower actually buys item $e$. In order to achieve this, the leader is allowed to \enquote{tempt} the follower, by offering so-called \emph{incentives}~$i(e)$ for each $e \in\U_L$ to the follower.
The leader then wants to solve
\begin{equation} \label{eq:knapsack-pricing}
  \begin{aligned}
    \max_{i \colon \U_L \to \R}\ & p_L(X^\star) - i(X^\star)\\
    \text{s.t. } & X^\star \in \argmax_{X \in \F}\ i(X) + p_F(X). \\
  \end{aligned}
\end{equation}
(For the sake of readability, we write $p_L(X^\star)$, $i(X^\star)$, and $p_F(X)$ instead of $p_L(X^\star \cap \U_L)$, $i(X^\star \cap \U_L)$, and $p_F(X \cap \U_F)$.)
The leader therefore has to consider the tradeoff between offering the follower good incentives $i(e)$ and still receiving a large payout $p_L(X^\star) - i(X^\star)$.
Even though we described this problem from the viewpoint of giving incentives, we will later see in \cref{sec:problem-definition} that an easy reformulation can equivalently frame the problem in terms of pricing. Therefore, as has also been observed by past researchers \cite{bohnlein2023stackelberg,bui2025solving}, problem~\eqref{eq:knapsack-pricing} fits into the framework of combinatorial pricing games.

\paragraph{The polynomial hierarchy.}
Many problems in bilevel optimization are expected to be even harder than $\NP$-hard -- they are naturally complete for the second level of the \emph{polynomial hierarchy}, the so-called complexity class $\Sigma^p_2$. 
If a problem is $\Sigma^p_2$-complete, then, under complexity-theoretic assumptions, it cannot be expressed as a mixed-integer linear formulation with only polynomially many variables and constraints. 
Therefore, the study and understanding of $\Sigma^p_2$-complete problems is important. 
We refer to Woeginger \cite{DBLP:journals/4or/Woeginger21} for an excellent introduction to $\Sigma^p_2$-completeness in bilevel optimization.
Past researchers have recognized that $\Sigma^p_2$-complexity likely applies to bilevel pricing problems. Indeed, Pferschy, Nicosia, Pacifici, and Schauer \cite{pferschy2021stackelberg} considered the knapsack pricing problem~\eqref{eq:knapsack-pricing}. 
They were able to show $\NP$-hardness as well as several other interesting approximability and inapproximability results for different variants of the follower's behavior. 
They could show that the variant of problem~\eqref{eq:knapsack-pricing}
where the leader can only choose between two different incentives per item
is $\Sigma^p_2$-hard. However, in the general case,
they could not resolve the complexity status and stated it as an open question whether the problem is $\Sigma^p_2$-complete.
Later, Böhnlein, Schaudt, and Schauer \cite{bohnlein2023stackelberg} considered the combinatorial pricing problem for multiple underlying problems. They obtained a polynomial-time algorithm for interval scheduling pricing and an inapproximaility result for independent set pricing. 
Furthermore, they showed that there exists some artificial, hand-crafted problem $\Pi$ such that $\Pi$ is in $\NP$ and the pricing problem for $\Pi$ is $\Sigma^p_2$-complete. Naturally, they remarked that such a result would be much more interesting for a more natural problem, and conjectured that both the knapsack pricing problem, and the set packing pricing problem (a generalization of independent set) are $\Sigma^p_2$-complete.

\paragraph{Our contribution.}
In this paper, we positively settle the open question and conjecture
from \cite{bohnlein2023stackelberg,pferschy2021stackelberg}
by proving that the knapsack pricing problem~\eqref{eq:knapsack-pricing} and the independent set pricing problem are $\Sigma^p_2$-complete.
Moreover, we are able to extend our techniques beyond a proof of the initial conjecture in order to prove a \emph{much more general} result: We show that the combinatorial pricing problem is in fact $\Sigma^p_2$-complete for \emph{every of the following problems}:
\begin{quote}
    knapsack,
    independent set,
    clique,
    vertex cover,
    set cover,
    hitting set,
    feedback vertex set,
    dominating set,
    feedback arc set,
    uncapacitated facility location,
    $p$-center,
    $p$-median,
    TSP,
    Steiner tree,
    3-dimensional matching,
    two-disjoint directed path,
    satisfiability,
    subset sum,
    Hamiltonian cycle,
\end{quote}

as well as many more problems (but at least 50) on an ever growing list of so-called \emph{SSP-complete} problems.
This is a subclass of $\NP$-complete problems introduced by Grüne and Wulf \cite{grune2025completeness,ReductionsNetwork}. 
The class contains most classically known $\NP$-complete problems, and is defined as all $\NP$-complete problems that satisfy a certain extra condition, more precisely explained in \cref{sec:framework}. 
This extra condition is on the one hand  general enough so that it seems to be satisfied by most \enquote{natural} $\NP$-complete problems, on the other hand it is restrictive enough to enable broad meta-theorems about $\Sigma^p_2$-completeness.
In fact, our main theorem states that, for every $\SSPNP$-complete problem $\Pi$, its pricing variant is $\Sigma^p_2$-complete.
Our result complements recent other meta-theorems about completeness in the polynomial hierarchy of multi-level optimization problems, regarding min-cost interdiction \cite{grune2025completeness}, min-max regret \cite{grune2025completeness}, unit cost interdiction \cite{grune2025unitcostinterdict}, and recoverable optimization \cite{grune2024recoverable}.

We remark that, for each individual problem from the list, it is not surprising that the pricing variant is $\Sigma^p_2$-complete. However, we argue that the main insight that our paper offers is that all these pricing problems are $\Sigma^p_2$-complete \enquote{for the same reason}.

\paragraph{Outline.} We quickly sketch the remainder of the paper, and our general proof strategy to obtain our meta-theorem. In \cref{sec:framework}, we summarize the main ideas of the SSP framework of Grüne and Wulf. 
In \cref{sec:problem-definition}, we define the pricing problem associated to an abstract combinatorial problem $\Pi$.
This definition needs to be extremely precise, since we later wish to argue over all such problems simultaneously. 
In \cref{sec:SAT-complete}, we first consider the pricing problem for only $\Pi = \textsc{Satisfiability}$ (\sat{}). One could argue that the pricing problem in case of \sat{} seems unnatural. 
However, the advantage of \sat{} is that, heuristically, it often seems well-suited to perform hardness proofs. Indeed, we can leverage this to show $\Sigma^p_2$-completeness of the \sat{} pricing problem in \cref{sec:SAT-complete}.
Despite the amenable properties of \sat{}, the proof is still involved.
Finally, in \cref{sec:meta-reduction}, we show how to make use of the strong structural insights of the SSP framework.
We show that the $\Sigma^p_2$-completeness of the \sat{} pricing problem extends to \emph{all} other problems in the class $\SSPNPc$.
This results in our main theorems, \cref{thm:meta-reduction-max,thm:meta-reduction-min,thm:meta-reduction-feas}.
We remark that our proof follows a natural approach: Instead of proving the $\Sigma^p_2$-completeness of the knapsack pricing problem directly, we first consider the case of \sat{}, which is easier to handle (at least according to the judgment of the authors of this article), and then extend our techniques from \sat{} to all other problems.

\section{Preliminaries}
\label{sec:preliminaries}

A \emph{language} is a set $L\subseteq \{0,1\}^*$.
A language $L$ is contained in $\Sigma^p_2$ if and only if there exists some polynomial-time computable function $V$ (also called \emph{verifier}) such that, for all $w \in \set{0,1}^*$ and suitable $m_1,m_2 = |w|^{O(1)}$,
$$
    w \in L \ \Leftrightarrow \ \exists y_1 \in \set{0,1}^{m_1} \ \forall y_2 \in \set{0,1}^{m_2}: V(w,y_1,y_2) = 1.
$$
An introduction to the polynomial hierarchy and the class $\Sigma^p_2$ can be found in the book by Papadimitriou \cite{DBLP:books/daglib/0072413} or in the article by Jeroslow \cite{DBLP:journals/mp/Jeroslow85}.
An introduction specifically in the context of bilevel optimization can be found in the article of Woeginger \cite{DBLP:journals/4or/Woeginger21}.

A \emph{many-one-reduction} or \emph{Karp-reduction} from a language $L$ to a language~$L'$ is a map $f : \{0,1\}^* \to \{0,1\}^*$ such that $w \in L$ if and only if $f(w) \in L'$ for all $w \in \{0,1\}^*$. 
A language $L$ is $\Sigma^p_2$-hard if every $L' \in \Sigma^p_2$ can be reduced to $L$ with a polynomial-time many-one reduction. If $L$ is both $\Sigma^p_2$-hard and contained in $\Sigma^p_2$, it is $\Sigma^p_2$-complete.

A \emph{Boolean variable} $x$ is a variable that takes one of the values 0 or 1. Let $X = \fromto{x_1}{x_n}$ be a set of variables. The corresponding \emph{literal set} is given by $L = \fromto{x_1}{x_n} \cup \fromto{\overline x_1}{\overline x_n}$. A \emph{clause} is a disjunction of literals.
A Boolean formula is in \emph{conjunctive normal form (CNF)} if it is a conjunction of clauses. It is in \emph{disjunctive normal form (DNF)} if it is a disjunction of conjunctions of literals.
In this paper, we use the notation where a clause is represented by a subset of the literals, and a CNF formula $\varphi$ is represented by a set of clauses. 
For example, the set $\set{\set{x_1, \overline x_2},\set{x_2, x_3}}$ corresponds to the formula $(x_1 \lor \overline x_2)\land (x_2 \lor x_3)$.
We write $\varphi(X)$ to indicate that formula $\varphi$ depends only on $X$.
Sometimes we are interested in cases where the variables are partitioned, we indicate this case by writing $\varphi(X,Y,\ldots)$.
An \emph{assignment} of variables is a map $\alpha : X \to \{0,1\}$.
The evaluation of the formula $\varphi$ under assignment $\alpha$ is denoted by $\varphi(\alpha) \in \{0,1\}$.
Given a formula $\varphi(X, Y)$ with partitioned variables and a partial assignment $\alpha$ on $X$, we also denote the remaining formula by $\varphi(\alpha, Y)$.

For some cost function $c : U \to \R$, and some subset $U' \subseteq U$, we define the cost of the subset $U'$ as $c(U') := \sum_{u \in U'} c(u)$. For a map $f : A \to B$ and some subset $A' \subseteq A$, we define the image of the subset $A'$ as $f(A') = \set{f(a) : a \in A'}$. 
\section{SSP Framework}
\label{sec:framework}

We repeat the most important concepts from \cite{grune2025completeness} necessary to understand this paper.

\begin{definition}[Linear Optimization Problem, from \cite{grune2025completeness}]
\label{def:LOSPP}
    A linear optimization problem (or short LOP problem)  $\Pi$ is a tuple $(\I, \U, \F, w, t)$ such that
    \begin{itemize}
        \item $\I \subseteq \{0,1\}^*$ is a language. We call $\I$ the set of instances of $\Pi$.
        \item To each instance $I \in \I$, there is some
        \begin{itemize}
            \item set $\U(I)$ called the universe associated to the instance $I$.
            \item set $\F(I) \subseteq 2^{\U(I)}$ called the feasible solution set associated to instance $I$. 
            \item function $w^{(I)}: \U(I) \rightarrow \Z$ mapping each universe element to its weight.
            \item threshold $t^{(I)} \in \Z$. 
        \end{itemize}
    \end{itemize}
\end{definition}

For $I \in \I$, we define the solution set $\sol(I) := \set{S \in \F(I) : w^{(I)}(S) \leq t^{(I)}}$ as the set of feasible solutions below the weight threshold. 
The instance $I$ is a yes-instance if and only if $\sol(I) \neq \emptyset$.
We assume that $|\U(I)| \leq \poly(|I|)$ and that it can be checked in polynomial time in $|I|$ whether some proposed set $F \subseteq \U(I)$ is feasible. The LOP problem is called a \emph{minimization LOP problem} if $w^{(I)},t^{(I)} \geq 0$ for all instances $I \in \I$. 
It is called a \emph{maximization LOP problem} if $w^{(I)}, t^{(I)} \leq 0$ for all instances $I \in \I$.
The following is an example of a minimization LOP problem (in this example, the function $w^{(I)}$ is always the unit weight function).

\vspace{.5\baselineskip}
\noindent\textsc{Vertex Cover}
\vspace{-0.67\baselineskip}
\begin{description}
\item{\it Instances:} Graph $G = (V, E)$, number $k \in \N$.
\item{\it Universe:} Vertex set $V =: \U$.
\item{\it Feasible solution set:} The set of all vertex covers of $G$.
\item{\it Solution set:} The set of all vertex covers of $G$ of size at most $k$.
\end{description}

It turns out that oftentimes the mathematical discussion is a lot clearer when one omits the concepts $\F(I)$, $w^{(I)}$, and $t^{(I)}$ since, for the abstract proofs of the theorems, only $\I$, $\U$, and $\sol$ are important.
This leads to the following abstraction from the concept of an LOP problem:

\begin{definition}[Subset Search Problem (SSP), from \cite{grune2025completeness}]
\label{def:SSP}
A subset search problem (or short SSP problem) $\Pi$ is a tuple $(\I, \U, \sol)$ such that
\begin{itemize}
    \item $\I \subseteq \set{0,1}^*$ is a language. We call $\I$ the set of instances of $\Pi$. 
    \item To each instance $I \in \I$, there is some set $\U(I)$, which we call the universe associated to the instance $I$. 
    \item To each instance $I \in \I$, there is some (potentially empty) set $\sol(I)\subseteq \powerset{\U(I)}$, which we call the solution set associated to the instance $I$.
\end{itemize}
\end{definition}

An instance of an SSP problem is a yes-instance if $\sol(I) \neq \emptyset$. 
Every LOP problem becomes an SSP problem with the definition $\sol(I) := \set{S \in \F(I) : w^{(I)}(S) \leq t^{(I)}}$.
We call this the \emph{SSP problem derived from an LOP problem}. Some problems are more naturally modeled as an SSP problem to begin with, rather than as an LOP problem.
For example, the satisfiability problem becomes an SSP problem with the following definition.

\vspace{.5\baselineskip}
\noindent\textsc{Satisfiability}
\vspace{-.67\baselineskip}
  \begin{description}
    \item{\it Instances:} Literal set $L = \fromto{\ell_1}{\ell_n} \cup \fromto{\overline \ell_1}{\overline \ell_n}$, clause set $C = \fromto{C_1}{C_m}$ such that $C_j \subseteq L$ for all $j \in \fromto{1}{m}.$
    \item{\it Universe:} $L =: \U$.
    \item{\it Solution set:} The set of all subsets $L' \subseteq \U$ of the literals such that $|L' \cap \set{\ell_i, \overline \ell_i}| = 1$ for all $i \in \fromto{1}{n}$ and $|L' \cap C_j| \geq 1$ for all clauses~$C_j \in C$.
  \end{description}
%\vspace{.5\baselineskip}

Grüne and Wulf introduce a new type of reduction, called \emph{SSP reduction}. 
Due to space reasons, we refer the reader to \cite{grune2025completeness} for an extensive explanation and examples of this concept.

\begin{definition}[SSP Reduction, from \cite{grune2025completeness}]
\label{def:ssp-reduction}
    Let $\Pi = (\I,\U,\sol)$ and $\Pi' = (\I',\U',\sol')$ be two SSP problems. We say that there is an SSP reduction from $\Pi$ to $\Pi'$, and write $\Pi \leqSSP \Pi'$, if
    \begin{itemize}
        \item There exists a function $g : \I \to \I'$, computable in polynomial time in the input size $|I|$, such that $I$ is a yes-instance if and only if $g(I)$ is a yes-instance (i.e., $\sol(I) \neq \emptyset$ if and only if $\sol'(g(I)) \neq \emptyset$).
        \item There exist functions $(f_I)_{I \in \I}$, computable in polynomial time in $|I|$, such that, for all instances $I \in \I$, we have that $f_I : \U(I) \to \U'(g(I))$ is an injective function mapping from the universe of the instance $I$ to the universe of the instance $g(I)$ such that 
        \begin{equation*}
            \set{f_I(S) : S \in \sol(I) } = \set{S' \cap f_I(\U(I)) : S' \in  \sol'(g(I))}.
        \end{equation*}

    \end{itemize}
\end{definition}

The class of $\SSPNP$-complete problems is denoted by $\SSPNPc$ and consists of all SSP problems $\Pi$ that are polynomial-time verifiable and such that $\textsc{Satisfiability} \leq_\text{SSP} \Pi$.
The main observation in~\cite{grune2025completeness} is that many classic problems are contained in the class $\SSPNPc$, 
and that this fact can be used to prove that their corresponding min-max versions are $\Sigma^p_2$-complete.
Finally, an LOP problem $\Pi$ is called \emph{$\LOP$-complete} if the SSP problem derived from it is $\SSPNP$-complete, 
and the SSP reduction used to achieve this has the additional property that $\set{F \in \F(I) : w^{(I)}(F) \leq t^{(I)} - 1} = \emptyset$, where $I$, $w^{(I)}$, and $t^{(I)}$ are the instance, the weights, and the threshold used in the reduction.

\paragraph{Example.} Garey and Johnson \cite{gareyjohnson1979} show a reduction from \sat{} to vertex cover. Given a \sat{} formula $\varphi$ with $n$ variables and $m$ clauses, they construct a graph $G_\varphi$ such that $G_\varphi$ has a vertex cover of size $n +2m$ if and only if $\varphi$ is satisfiable.
The graph $G_\varphi$ contains a subset $V'$ of $2n$ vertices that \enquote{encode} the \sat{} problem: Given a vertex cover of size $n + 2m$, 
if we restrict it to $V'$, we directly obtain a satisfying assignment. On the other hand, given a satisfying assignment, we can easily translate it to $V'$ and complete this partial vertex cover to a vertex cover of size $n+2m$. 
Finally, $G_\varphi$ has no vertex covers of size less than $n+2m =: t^{(I)}.$ This proves that the vertex cover problem is $\LOP$-complete.

\section{Problem Definition}
\label{sec:problem-definition}

Our problem definition is slightly different for minimization and maximization problems.
First, let $\Pi = (\I, \U, \F, w, t)$ be a maximization LOP problem as defined in the previous section. 
We now formally define the problem $\textsc{Pricing-$\Pi$}$. 
The input is a tuple $(I, \U_L, \U_F, p_L, p_F)$ such that $I \in \I$, $\U(I) = \U_L \cup \U_F$, $\U_L \cap \U_F = \emptyset$, $p_L : \U_L \to \Zgeq$, and $p_F : \U_F \to \Zgeq$. The problem is to compute the value of
\begin{equation} \label{eq:pricing-pi-max-1}
  \begin{aligned}
    \max_{i \colon \U_L \to \R}\ & p_L(X^\star) - i(X^\star)\\
    \text{s.t. } & X^\star \in \argmax_{X \in \F(I)}\ i(X) + p_F(X). \\
  \end{aligned}
\end{equation}

Since we are concerned with computational complexity, we want to treat the problem as a decision problem. We therefore assume that we are given an additional threshold $T_\text{pr} \in \Qgeq$, and the question is whether the value of \eqref{eq:pricing-pi-max-1} is at least $T_\text{pr}$.
The problem $\textsc{Pricing-$\Pi$}$ is then the set of all such decision problems for all possible tuples $(I, \U_L, \U_F, p_L, p_F)$ with properties as specified above, and all possible thresholds $T_\text{pr} \in \Qgeq$. 

As it is often done in bilevel optimization, we follow the optimistic assumption: 
If the follower can choose between different solutions $X$ that have the same maximal follower's objective value $i(X)+ p_F(X)$, we assume optimistically that they will select the one with the best objective value $p_L(X) - i(X)$ for the leader.

In the remainder of the paper, it will be more convenient to instead consider an equivalent reformulation of \eqref{eq:pricing-pi-max-1}. 
Let us define for all items $e \in \U_L$ the term $d(e) := p_L(e) - i(e)$, and let $p(e) := p_L(e)$ for $e \in \U_L$ and $p(e) := p_F(e)$ for $e \in \U_F$.
Then, since the function $p_L : \U_L \to \Zgeq$ is a fixed function that is part of the input, we have that optimizing over $i : \U_L \to \R$ is equivalent to optimizing over $d : \U_L \to \R$. Therefore, \eqref{eq:pricing-pi-max-1} is equivalent to

\begin{equation} \label{eq:pricing-pi-max-2}\tag{$\text{P}_{\text{max}}$}
  \begin{aligned}
    \max_{d \colon \U_L \to \R}\ & d(X^\star)\\
    \text{s.t. } & X^\star \in \argmax_{X \in \F(I)}\ p(X) - d(X). \\
  \end{aligned}
\end{equation}
Observe that \eqref{eq:pricing-pi-max-2} can be understood as a problem where the follower has a gross profit $p(e) \geq 0$ for every item $e \in \U_L \cup \U_F$, and is aiming to maximize their profit.
The leader can place prices on the items in $\U_L$, which increases the leader's profit, while decreasing the follower's profit. 
Note that, if $X \cap \U_L \neq \emptyset$ for all $X \in \F(I)$, then the leader can raise the price to infinity since the follower has to include such leader's items into a feasible solution.
Therefore, \eqref{eq:pricing-pi-max-2} is unbounded. 
In the remainder of the paper, we assume that this is not the case, that is, we make sure that there is always at least one feasible solution $F \subseteq \U_F$.

Now, let $\Pi = (\I, \U, \F, w, t)$ be a minimization LOP problem as described in the previous section.
Each instance of $\Pi$ can be described as $\min_{X \in \F(I)}w(X) = \max_{X \in \F(I)}-w(X)$ for some weight function $w \geq 0$. It is therefore natural to define the  problem $\textsc{Pricing-$\Pi$}$ in this case as 
\begin{equation} \label{eq:pricing-pi-min}\tag{$\text{P}_{\text{min}}$}
  \begin{aligned}
    \max_{d \colon \U_L \to \R}\ & d(X^\star)\\
    \text{s.t. } & X^\star \in \argmin_{X \in \F(I)}\ c(X) + d(X). \\
  \end{aligned}
\end{equation}
for a function $c : \U(I) \to \Z_{\geq 0}$. 
Problem \eqref{eq:pricing-pi-min} can be interpreted as a pricing problem, where the follower tries to pay minimal cost $c(X) + d(X)$ 
and the leader can increase the prices for their part of the universe, trying to make attractive offers to the follower while still aiming for a high profit.

\paragraph{Difference between the domains $\R$ and $\Rgeq$.} 
In both problems \eqref{eq:pricing-pi-max-2} and~\eqref{eq:pricing-pi-min}, the leader is unconstrained, i.e., they can select any $d : \U_L \to \R$. 
Other natural choices, which have also been considered by past authors \cite{bui2025solving,pferschy2021stackelberg}, are to only allow $d(e) \geq 0$ or $d(e) \leq p(e)$ (which is equivalent to $i(e) \geq 0$ in \eqref{eq:pricing-pi-max-1}) for all $e \in \U_L$.
For our exposition, it is most convenient to consider $d : \U_L \to \R$. 
However, we later show that this does not make a difference for our main result, i.e., all the other natural choices lead to $\Sigma^p_2$-complete problems as~well. 

\paragraph{Example.} If $\Pi = \text{knapsack}$, then $\Pi$ is a maximization problem, and our problem $\textsc{Pricing-$\Pi$}$ is given by $\eqref{eq:pricing-pi-max-2}$, 
where $p : \U(I) \to \Z_{\geq 0}$ are the profits and $\F(I)$ is the family of sets of items of total weight below the weight threshold. By our explanations, this is equivalent to the Stackelberg knapsack problem~\eqref{eq:knapsack-pricing} from~\cite{pferschy2021stackelberg}.
If $\Pi = \text{vertex cover}$, then $\Pi$ is a minimization problem, and $\textsc{Pricing-$\Pi$}$ is given by \eqref{eq:pricing-pi-min}, 
where $c : \U(I) \to \Zgeq$ are the vertex weights and $\F(I)$ is the family of all vertex covers of the input graph.

\paragraph{Feasibility pricing.}
Finally, we consider problems such as $\Pi =$ \textsc{Satis\-fiab\-il\-ity}, which are neither maximization nor minimization problems, as they are concerned only with feasibility. 
As explained in the introduction, we wish to prove our meta-theorem starting from problems of this kind. Let $\Pi = (\I, \U, \sol)$ be an SSP problem as described in the previous section. Given a tuple $(I, \U_L, \U_F, p)$ with $p :\U \to \Z_{\geq 0}$, we define the problem \textsc{Feas-Pricing-$\Pi$} as the decision problem corresponding to
\begin{equation} \label{eq:pricing-pi-feas}\tag{$\text{P}_{\text{feas}}$}
  \begin{aligned}
    \max_{d \colon \U_L \to \R}\ & d(X^\star)\\
    \text{s.t. } & X^\star \in \argmax_{X \in \sol(I)}\ p(X) - d(X). \\
  \end{aligned}
\end{equation}

\begin{theorem}
  For all LOP problems $\Pi$, we have $\textsc{Pricing-$\Pi$} \in \Sigma^p_2$.
  Also for all $\SSPNP$ problems $\Pi$, we have $\textsc{Feas-Pricing-$\Pi$} \in \Sigma^p_2$.
\end{theorem}
\begin{proof}
	In order to prove the containment in $\Sigma^p_2$, we would like to find a witness $\exists y_1 \forall y_2$ of polynomial size that can be verified by a polynomial-time computable function~$V$ with input $y_1$, $y_2$ and the instance of the problem.
	Consider w.l.o.g. only problem~\eqref{eq:pricing-pi-max-2}.
	The instance of the corresponding decision problem \textsc{Pricing-$\Pi$} $= (I, \U_L, \U_F, p, T_\text{pr})$ consists of an instance $I$ as given in the underlying problem~$\Pi$, leader elements $\U_L \subseteq \U(I)$, follower elements $\U_F \subseteq \U(I)$, a profit function $p: \U \rightarrow \Zgeq$, and a threshold $T_\text{pr} \in \Qgeq$.
	Let $n := |\U(I)|$ and $n_L := |\U_L|$.
	A possible witness-verifier pair to this problem can be defined as follows:
	\begin{equation}\label{sigma2-containment-formulation}
		\begin{aligned}
			&\exists d : \U_L \to \R, \exists X^\star \in \F(I) : \forall X \in \F(I) : \\
			&\qquad\qquad \left[p(X^\star) - d(X^\star) \geq p(X) - d(X)\right] \land d(X^\star) \geq T_\text{pr}.
		\end{aligned}
              \end{equation}
	Since $d : \U_L \to \R$ involves real numbers, we have to worry about their encoding length.
	Thus,  it is not entirely obvious why the witness is of polynomial size and can be checked efficiently.
	With the following polyhedral arguments, however, we can show that $d$ can be encoded in polynomial size.
	For this, let $x_i \in \set{0,1}^n$ be the characteristic vector of the $i$-th feasible solution $X_i \in \F(I)$.
	Since $\Pi \in \SSPNP$,
        each $x_i$ is a
        string of length $n = |\U(I)|$. However, there may be exponentially many such $x_i$'s.
	Then, \eqref{sigma2-containment-formulation} is true if and only if the following nonlinear mixed-integer program contains a feasible point.
	\begin{equation} \label{eq:mip-formulation}
		\begin{aligned}
                        \min~ & 0 \\
			\text{s.t. }
			& d \in \R^{n_L} \\
			& x^\star \in \F(I)\\
			& p^tx^\star - (d,0)^tx^\star \geq p^tx_i - (d,0)^tx_i & \forall x_i \in \F(I)\\
                            & (d,0)^tx^\star \geq T_{\text{pr}}
		\end{aligned}
	\end{equation}	
	Note that, in this formulation, we let the leader choose the characteristic vector~$x^\star$ to the feasible solution $X^\star \in \F(I)$ that maximizes the follower's objective.
	This is possible due to the optimistic assumption.
	Furthermore, the constraint \enquote{$x^\star \in \F(I)$} can in fact be expressed as a MIP constraint because we have assumed $\Pi$ to be a problem in $\SSPNP$.
	More precisely, the feasibility of $\Pi$ can be checked in polynomial time,
	and this Turing machine can be translated into polynomially many MIP constraints (by the Cook-Levin theorem and a corresponding reduction from {\sc Satisfiability} to MIP).
	
	If we have a yes-instance to the problem \textsc{Pricing-$\Pi$}, there is a feasible solution $(d,x^\star)$ to \eqref{eq:mip-formulation}.	
	Then, $x^\star$ encodes an optimal follower's solution $X^\star \in \F(I)$ corresponding to the leader's solution $d$.
	Since $\exists X^\star \in \F(I)$ is on the same quantifier level as $\exists d : \U_L \to \R$, we can fix the characteristic vector $x^\star$.
	Then, let $P_{x^\star}$ denote the resulting polyhedron associated to the linear program where we fix $x^\star$ and vary $d \in \R^{n_L}$ in \eqref{eq:mip-formulation}.
	We remark that this linear program may have exponentially many constraints and is not necessarily bounded. However, $P_{x^\star}$ is nonempty since $d \in P_{x^\star}$.
        Let $L_{\max}$ be the largest size of a number in the linear inequality system of \eqref{eq:mip-formulation}
        (where the bit-size of a rational number $r/q$ is given by $\log(|r|) + \log(|q|) + 1$).
         Furthermore, $P_{x^\star} \subseteq \R^{n_L}$
         and the size of
         each constraint is bounded by $2 n_L L_{\max} + 2n L_{\max} \leq 4n L_{\max} =: m$ (note that $x^\star, x_i \in \{0,1\}^{n}$ for all $x_i \in \F(I)$).
	
	We now claim that, for all thresholds $T_{\text{pr}}$, the following holds:
	If there is a leader's solution $d$ with profit at least $T_{\text{pr}}$, then there is also a leader's solution $\widehat d$ with profit at least $T_{\text{pr}}$ such that $\widehat d$ has a bit-size of at most $16n^3L_{\max}$.
	For the proof, we use a theorem of Schrijver.

\textit{Fact.} (Theorem 10.2 in \cite{schrijver1998theory}) For any rational $N$-dimensional polyhedron, its vertex complexity is bounded by $4N^2$ times its facet complexity.
	
	An upper bound to the facet complexity of $P_{x^\star}$ is given by $m = 4n L_{\max}$.
	This induces an upper bound to the vertex complexity of $P_{x^\star}$ of $4n_L^2m \le 16n^3 L_{\max}$.
	Then, $P_{x^\star}$ can be described by vectors $a_1, \ldots, a_k, b_1, \ldots, b_t$, each of bit-size at most $16n^3 L_{\max}$, with
	$P_{x^\star} = \conv\{a_1, \ldots, a_k\} + \cone\{b_1, \ldots, b_t\}$.
	Since $P_{x^\star} \neq \emptyset$ and any point in $P_{x^\star}$ is a feasible point in which the leader receives a profit of at least $T_{\text{pr}}$,
         we can set $\widehat d = a_1$.
	Thus, there is a suitable $\widehat d$ of bit-size at most $16n^3 L_{\max}$, which shows the claim.	
	
	Note that the encoding of $\widehat d$ contains the encodings of the vector components as rational numbers. 
	It follows that there are an $x^\star \in \{0,1\}^n$, encoding an $X^\star \in \F(I)$,
        and a $\widehat d: \U_L \to \R$ of polynomial size such that for all $x \in \{0,1\}^n$, encoding $X \in \F(I)$, $\left[p(X^\star) - \widehat d(X^\star) \geq p(X) - \widehat d(X)\right] \land \widehat d(X^\star) \geq T_\text{pr}$ holds. \qed
\end{proof}

\section{Feasibility-Pricing-\sat{} is $\Sigma^p_2$-complete}
\label{sec:SAT-complete}

In this section, we consider \textsc{Feas-Pricing-\sat{}}. We are given a Boolean formula~$I$ in \emph{conjunctive} normal form over the literals $\U(I) := \set{\ell_1,\dots,\ell_N} \cup \set{\overline \ell_1,\dots,\overline \ell_N}$, 
where the set $\sol(I) \subseteq 2^{\U(I)}$ of satisfying subsets is defined as in \cref{sec:framework}.
Moreover, we are given a partition $\U(I) = \U_L \cup \U_F$ (with $\U_L \cap \U_F = \emptyset$),
a function $p : \U(I) \to \Zgeq$,
and a threshold $T_{\text{pr}} \in \Qgeq$.
The problem then is to decide whether the optimal value of \eqref{eq:pricing-pi-feas} is at least $T_{\text{pr}}$.

We give a sketch of the proof idea
and remark that some of our arguments are based on the elegant proof of Böhnlein et al. \cite{bohnlein2023stackelberg}.
We reduce from the canonical $\Sigma^p_2$-complete problem $\exists \forall$-DNF-\sat{}. 
Here, we are given a \sat{} formula $\varphi(A,B)$ in \emph{disjunctive} normal form, 
and we are asked to evaluate whether there exists an assignment of the $A$-variables such that, for all assignments of the $B$-variables, $\varphi(A,B)$ is true.
We now want to transform the DNF formula $\varphi$ into a new CNF formula $I$, and find $\U_L \cup \U_F = \U(I)$, $p$, and $T_{\text{pr}}$ such that the decision version of \eqref{eq:pricing-pi-feas} becomes difficult.
In particular, the idea is that the choice of $d : \U_L \to \R$ forces the follower to select elements that correspond to the assignment of~$A$, and residual elements from
the follower's solution
should correspond to the assignment of $B$.
The first idea is that our formula $I$ begins with a clause $(h_1 \lor h_2)$, where $h_1 \in \U_F$. 
If the leader sets the prices too high, the follower simply chooses $h_1$ and obtains a fixed profit.
In the other case, the follower chooses $h_2$ and the following considerations apply: 
The instance $I$ contains variables $v^t_i,v^f_i,v^o_i$ for $i \in \fromto{1}{n}$, where $v_i^t,v_i^f \in \U_L$ and $v_i^o \in \U_F$.
The idea is that the leader will set prices for $v_i^t$ and $v_i^f$, 
and the follower will generally prefer the element that maximizes their profit, i.e., the element with the smaller price. 
This implicitly determines an assignment of $A$, by having $a_i = 1$ if $d(v_i^t) < d(v_i^f)$ and $a_i=0$ otherwise. 
However, if the prices are too high, the follower simply switches to~$v_i^o$.
Finally, there is a pair of variables $z_1, z_2$ such that the solution contains exactly one of them. 
The leader would like to collect prices on $z_1$, but the value of $z_1$ is linked to the satisfiability of $\exists A \forall B \varphi(A,B)$.
If this formula is false, the follower can evade the leader's prices.
In summary, we are able to show that the leader can collect high prices if and only if $\exists A \forall B \varphi(A,B)$.

\begin{theorem}
\label{thm:pricing-SAT}
\textsc{Feas-Pricing-\sat{}} is $\Sigma^p_2$-hard.
\end{theorem}
\begin{proof}
Given an instance $\exists A \forall B \varphi(A,B)$ of $\exists \forall$-DNF-SAT with $|A| = |B| = n$, we define the following instance $(I, \U_L, \U_F, p)$ of \textsc{Feas-Pricing-SAT}. 
Let $M := 2n$, let $\texttt{once}(y_1,y_2,y_3)$ be the CNF-\sat{} formula that is true if and only if exactly one of $y_1$, $y_2$, and $y_3$ is true, and let $\oplus$ denote the XOR relation.
\renewcommand{\arraystretch}{1.5}
\begin{tabular}{ll}
  {\it Variables:}\ \ & $V := \set{h_1, h_2, z_1,z_2} \cup \bigcup_{i=1}^n \set{v^t_i, v^f_i, v^o_i, a_i, b_i}$ \\
  {\it Literals:}\ \ &  $\U(I) = V \cup \overline V$ with $\U_L = \bigcup_{i=1}^n \set{v^t_i, v^f_i} \cup \set{z_1},\ \ \U_F = \U(I) \setminus \U_L$ \\
  {\it Profits ($\U_L$):}\ \ & $p(v^t_i) = p(v^f_i) = M \ \forall i \in \fromto{1}{n},\ \ p(z_1) = M$\\
  {\it Profits ($\U_F$):}\ \ & $p(v^o_i) = 1\ \forall i \in \fromto{1}{n},\ \ p(h_1) = n,\ \ p(z_2) = 1,\ \ p(e) = 0 \text{ else}$
\end{tabular}

\noindent%
\begin{tabular}{ll}
  {\it Formula:}& $(h_1 \lor h_2)\land\left(h_1 \rightarrow (\bigwedge_{x \in V \setminus \set{h_1}}\overline x)\right)$\\
  &$\land \left(h_2 \rightarrow ((z_1 \oplus z_2)\land(z_2 \rightarrow \neg\varphi(A,B)))\right)$\\ 
  &$\land \left(h_2 \rightarrow \bigwedge_{i=1}^n\left( \texttt{once}(v_i^t,v_i^f,v_i^o) \land (v_i^t \rightarrow a_i) \land (v_i^f \rightarrow \overline a_i) \land (v_i^o \rightarrow z_1) \right) \right)$
\end{tabular}
This completes the description of the instance. It is easily seen that $I$ is in CNF, using that $\varphi$ is in DNF and De Morgan's law.
Let $k^\star := (n+1)M - n$. 
We claim that the optimal value of \eqref{eq:pricing-pi-feas} is at least $k^\star$ if and only if $\exists A \forall B \varphi(A,B)$ is true.

\paragraph{\enquote{If} direction.}
Assume that there exists an assignment $\alpha : A \to \set{0,1}$ that makes $\exists A \forall B \varphi(A,B)$ true.
We define a pricing function $d : \U_L \to \R$ by having $d(z_1) = M$ and, for all $i \in \fromto{1}{n}$,
\begin{align*}
&d(v^t_i) = M-1, &&d(v^f_i) = M && \text{ if }\alpha(a_i) = 1,\\
&d(v^t_i) = M, &&d(v^f_i) = M - 1 && \text{ if }\alpha(a_i) = 0.
\end{align*}
We claim that this choice of $d$ enables the leader to obtain $k^\star$ profit.
First, we argue that the follower's optimal value is $n$.

If the follower chooses $h_1$, their profit is exactly $n$, since $h_1$ implies that, for all other variables, the negative literal is selected, and only positive literals induce profit for the follower.
Note that this is independent of the choice of $d$.

Otherwise, the follower chooses $h_2$.
Then, they obtain at most a profit of~$1$ from every clause $\texttt{once}(v_i^t,v_i^f,v_i^o)$. 
Furthermore, they cannot obtain profit from~$z_1$, since $d(z_1) = p(z_1) = M$.
At last, we have to consider the element~$z_2$.
For this, suppose that the follower collects a profit of $1$ from $z_2$.
Then, due to $v_i^o \rightarrow z_1$ and $z_1 \oplus z_2$, the follower's solution $X$ does not contain any $v_i^o$, but contains either $v_i^t$ or $v_i^f$ for all~$i \in \fromto{1}{n}$.
If $X$ contains exactly the element from $\{v_i^t, v_i^f\}$ that has cheaper prices $d$, for all $i \in \fromto{1}{n}$,
then, due to $v_i^t \rightarrow a_i$ and $v_i^f \rightarrow \overline a_i$, the variables $A$ have to be set according to the assignment $\alpha$.
Then the clause $z_2 \rightarrow \neg\varphi(A,B)$
and the fact that $\alpha$ makes $\exists A \forall B \varphi(A,B)$ true
imply $z_2 \not\in X$.
If $X$ contains at least one element $v_i^t$ or $v_i^f$ with leader's price $M$, then the total profit obtained from the clauses $\texttt{once}(v_i^t,v_i^f,v_i^o)$ is at most $n-1$.
In all cases, the total follower's profit is at most $n$.
In total, the follower's optimal value is indeed $n$.
Finally, observe that the follower's solution 
\[
	X =\set{\overline h_1, h_2, z_1, \overline z_2} \cup \bigcup_{i=1}^n \set{\overline v^o_i} \cup \bigcup_{\alpha(a_i) = 1} \set{v^t_i, \overline v^f_i, a_i} \cup \bigcup_{\alpha(a_i) = 0} \set{\overline v^t_i, v^f_i, \overline a_i}
\] 
achieves the follower's optimal value $n$ and has $d(X) = k^\star$.
Due to the optimistic assumption, the leader can make at least $k^\star$ profit.

\paragraph{\enquote{Only if} direction.}
Assume that there exists $d : \U_L \to \R$ such that $d$ achieves profit at least $k^\star$ for the leader. 
We need to show that $\exists A \forall B \varphi(A,B)$ is satisfiable.
Let $X_d$ be an optimistic follower's response to $d$, i.e., $d(X_d) \geq k^\star$ and $p(X_d) - d(X_d) = \text{OPT}_F$, where $\text{OPT}_F$ is the follower's optimal value.

First, we show that $\text{OPT}_F = n$.
The follower can always choose $h_1$, hence, $\text{OPT}_F \geq n$.
By construction of the instance, for all $X \in \sol(I)$, we have $p(X) \leq (n+1)M$.
Therefore, the follower's profit is $\text{OPT}_F = p(X_d) - d(X_d) \leq (n+1)M - k^\star = n$.
It follows that $\text{OPT}_F = n$.

Second, if the follower chooses $h_1$, then the leader's profit is $0$, since $h_1$ implies that, for all other variables, the negative literal is selected, and only positive literals induce profit for the leader.
Hence, $h_1 \notin X_d$ and $h_2 \in X_d$ because the leader's profit is at least $k^\star > 0$.

Third, observe that $d(z_1) \geq M$.
Indeed, if $d(z_1) < M$,
then the follower can pick $\set{v^o_1,\dots,v^o_n} \cup \set{z_1}$ as a partial solution and obtain strictly more than $n$ profit, contradicting $\text{OPT}_F = n$.

Fourth, we claim that both $z_1 \in X_d$ and $|X_d \cap \set{v_i^t,v_i^f}| = 1$ for all $i \in \fromto{1}{n}$.
Indeed, assume otherwise, then by the structure of the instance, $X_d$ contains at most $n$ elements $e$ with $p(e) = M$.
Therefore, $p(X_d) \leq nM + 1$.
But, since $d(X_d) \geq k^\star$, we have $p(X_d) - d(X_d) \leq nM + 1 - k^\star = (n+1) - M \leq 0$, which is a contradiction to $\text{OPT}_F \geq n$.

We claim that these four properties implicitly define a satisfying assignment to $\exists A \forall B \varphi(A,B)$ by setting $\alpha(a_i) = 1$ if $v^t_i \in X_d$ (and $a_i \in X_d$) and $\alpha(a_i) = 0$ if $v^f_i \in X_d$ (and $\overline a_i \in X_d$).
For this, assume that $\alpha$ does not satisfy the formula.
Then the follower can choose an assignment of $b_1,\dots,b_n$ that makes $\varphi(\alpha, B)$ false. Then $X' := X_d \setminus \set{z_1, \overline z_2} \cup \set{\overline z_1, z_2}$ is also a valid solution to $I$. 
But since $d(z_1) \geq M$ and $p(z_2) = 1$, the follower's profit of $X'$ is strictly greater than the profit of $X_d$, a contradiction to the fact that $X_d$ is an optimal follower's solution.
We conclude that $\exists A \forall B \varphi(A,B)$ is true. \qed
\end{proof}

\begin{corollary}
\label{cor:d-well-behaved-SAT}
\textsc{Feas-Pricing-\sat{}} remains $\Sigma^p_2$-hard, even if
$d : \U_L \to \R$ in \eqref{eq:pricing-pi-feas} is
restricted to $0 \leq d$ or $d \leq p$ or $0 \leq d \leq p$.
\end{corollary}
\begin{proof}
We check that both directions of the proof of \cref{thm:pricing-SAT} remain true. 
For the \enquote{if} direction, observe that the function $d$ defined there already has the property $0 \leq d \leq p$.
For the \enquote{only if} direction, the assumption that a function $d$ with
any of the three restrictions
exists
is stronger than only assuming $d : \U_L \to \R$.
Clearly, starting with a stronger assumption, the reasoning remains correct.
\qed
\end{proof}

\section{Proof of the Main Theorem}
\label{sec:meta-reduction}
We prove our main theorem separately for maximization and minimization problems,
since there are subtle differences between the two cases. We remark that a joint proof would also be possible if one introduced more abstraction, but this would come at the cost of readability.
We are now ready to prove our main theorem.
\begin{theorem}
\label{thm:meta-reduction-max}
For any problem $\Pi$ that is an $\LOP$-complete maximization LOP problem (see \cref{sec:framework}), the problem \textsc{Pricing-$\Pi$} as in \eqref{eq:pricing-pi-max-2} is $\Sigma^p_2$-hard.
\end{theorem}
In the remainder of this section, we prove \cref{thm:meta-reduction-max}.
For this, let an instance $\exists A \forall B \varphi(A,B)$ of the $\Sigma^p_2$-complete problem $\exists \forall$-DNF-\sat{} be given.
By \cref{thm:pricing-SAT}, we can find in polynomial time a new Boolean formula $I_\text{\sat{}}$, a partition $\U(I_\text{\sat{}}) = \U_L \cup \U_F$, 
profits $p : \U_L \cup \U_F \to \Z_{\geq 0}$, and a number $k^\star \in \Zgeq$ such that $\exists A \forall B \varphi(A,B)$ is true if and only if the following \textsc{Feas-Pricing-\sat{}} instance has an optimal value of at least $k^\star$.
\begin{equation} \label{eq:SAT-meta-thm}
  \begin{aligned}
    \max_{d \colon \U_L \to \R}\ & d(X^\star)\\
    \text{s.t. } & X^\star \in \argmax_{X \in \sol(I_\text{\sat{}})}\ p(X) - d(X). \\
  \end{aligned}
\end{equation}

Furthermore, by \cref{cor:d-well-behaved-SAT}, this fact remains true even if the constraint $0 \leq d \leq p$ is added to \eqref{eq:SAT-meta-thm}.
In order to prove \cref{thm:meta-reduction-max}, we give a reduction from \textsc{Feas-Pricing-\sat{}} to \textsc{Pricing-$\Pi$} for the abstract problem $\Pi$.
We make use of the fact that $\Pi$ is $\LOP$-complete.
This implies that, given the instance $I_\text{\sat{}}$ of \sat{}, one can compute in polynomial time a tuple $(I_\Pi, w^{(I_\Pi)}, t^{(I_\Pi)})$,
where $I_\Pi$ is an instance of $\Pi$ with universe $\U(I_\Pi)$, set of feasible solutions $\F(I_\Pi)$, weight function $w^{(I_\Pi)} : \U(I_\Pi) \to \Z_{\geq 0}$, and threshold $t^{(I_\Pi)} \in \Zgeq$.
Note that a maximization LOP problem can be expressed by non-negative weights $w^{(I_\Pi)}$ and threshold $t^{(I_\Pi)}$ by defining $\sol(I_\Pi) = \set{F \in \F(I_\Pi) : w^{(I_\Pi)}(F) \geq t^{(I_\Pi)}}$.
In order to simplify the notation, we write $\U_\text{\sat{}}$, $\U_\Pi$, $\F$, $w$, and $t$ from now on for these objects (they depend on the specific instance of \sat{} or $\Pi$, but the instance will be fixed throughout this proof). 
By the definition of $\LOP$-completeness, the instances of \sat{} and $\Pi$ are related by the SSP property: There is a (polynomial-time computable) injective function $f : \U_\text{\sat{}} \to \U_\Pi$ and
\begin{align*}
\sol(I_\Pi) := \set{F \in \F &: w(F) \geq t} = \set{F \in \F : w(F) = t}\\
\text{ such that} \quad &\forall S' \in \sol(I_\Pi) : f^{-1}(S') \in \sol(I_\text{\sat{}})\\
\text{ and} \quad &\forall S \in \sol(I_\text{\sat{}}): \exists S' \in \sol(I_\Pi) : S' \cap f(\U_\text{\sat{}}) = f(S).
\end{align*}

Based on this abstract reduction from \sat{} to $\Pi$, we now define an instance of \textsc{Pricing-$\Pi$}. The underlying instance of $\Pi$ is $I_\Pi$. The partition of the universe $\U_\Pi = \U'_L \cup \U'_F$ is sketched in \cref{fig:schematic-meta-reduction} and given by
\begin{align*}
\U'_L &:= f(\U_L)\\
\U'_F &:= f(\U_F) \cup (\U_\Pi \setminus f(\U_\text{\sat{}})).
\end{align*}

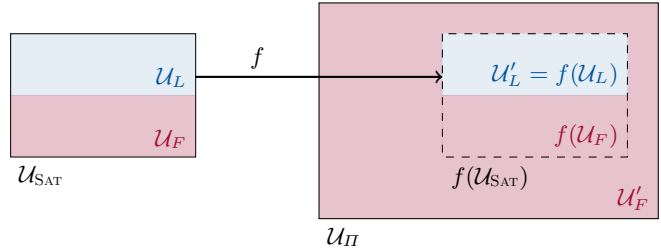
\begin{figure}
  \centering
  \begin{tikzpicture}[scale=0.82]
    \draw[] ($(-5,0)$) rectangle ($(-2,-2)$);
    \node[below right] at (-5,-2) {$\U_{\textsc{\sat{}}}$};
    \draw[lcolor, opacity=0, fill=lcolor, fill opacity=0.10] ($(-5,0)$) rectangle ($(-2,-1)$);
    \node[above left] at (-2,-1) {\color{lcolor}$\U_L$};
    \draw[fcolor, opacity=0, fill=fcolor, fill opacity=0.25] ($(-5,-1)$) rectangle ($(-2,-2)$);
    \node[above left] at (-2,-2) {\color{fcolor}$\U_F$};
    
    \draw[fill=fcolor, fill opacity=0.25] ($(0,0.5)$) rectangle ($(5.5,-3)$);
    \draw[dashed, fill=white] ($(2,0)$) rectangle ($(5,-2)$);
    \node[below right] at (2,-2) {$f(\U_{\textsc{\sat{}}})$};
    \node[below right] at (0,-3) {$\U_\Pi$};
    \draw[lcolor, opacity=0, fill=lcolor, fill opacity=0.10] ($(2,0)$) rectangle ($(5,-1)$);
    \node[above left] at (5,-1) {\color{lcolor}$\U'_L = f(\U_L)$};
    \draw[fcolor, opacity=0, fill=fcolor, fill opacity=0.25] ($(2,-1)$) rectangle ($(5,-2)$);
    \node[above left] at (5,-2) {\color{fcolor}$f(\U_F)$};
    \node[above left] at (5.5,-3) {\color{fcolor}$\U'_F$};
    
    \draw[->, thick] (-2,-0.7) -- (2,-0.7);
    \node[above] () at (-1,-0.7) {$f$};
  \end{tikzpicture}
  \caption{
    The relation between the universes when applying the meta-reduction from \textsc{Feas-Pricing-\sat{}} and \textsc{Pricing-$\Pi$}.
    The function $f$ maintains a one-to-one correspondence between the elements of $\U_L$ and $f(\U_L)$ as well as $\U_F$ and $f(\U_F)$.
  }
  \label{fig:schematic-meta-reduction}
\end{figure}

Let $n := |\U_\text{\sat{}}|$ and define 
$M := 4n\sum_{e \in \U_\text{\sat{}}}p(e)$.
We can think of $M$ as a very large value, i.e.,  $M \gg p(e)$ for all $e \in \U_\text{\sat{}}$.
Based on this, we define a new profit function $p' : \U'_L \cup \U'_F \to \Z_{\geq 0}$ by 
\[
p'(e) := 
\begin{cases}
Mw(e) + p(f^{-1}(e)) &\text{ if }e \in f(\U_\text{\sat{}})\\
Mw(e)				&\text{ if }e \in \U_\Pi \setminus f(\U_\text{\sat{}}).\\
\end{cases}
\]
In total, these definitions give rise to the following instance of $\textsc{Pricing-$\Pi$}$:

\begin{equation} \label{eq:pi-meta-thm}
  \begin{aligned}
    \max_{d' \colon \U'_L \to \R}\ & d'(Y^\star)\\
    \text{s.t. } & Y^\star \in \argmax_{Y \in \F}\ p'(Y) - d'(Y) \\
  \end{aligned}
\end{equation}
Observe that $p' \geq 0$, since $M,w \geq 0$, and so \eqref{eq:pi-meta-thm} is a valid instance of $\textsc{Pricing-$\Pi$}$.
In the remainder of this section, we show that $\exists A \forall B \varphi(A, B)$ is true if and only if both \eqref{eq:SAT-meta-thm} and \eqref{eq:pi-meta-thm} have value at least $k^\star$.
Indeed, if we can show this, we have established a polynomial-time reduction from $\exists \forall$-DNF-\sat{} to $\textsc{Pricing-$\Pi$}$, showing that $\textsc{Pricing-$\Pi$}$ is $\Sigma^p_2$-hard.
First, let $\alpha_\Pi$ be the optimal value of $I_\Pi$,~i.e.,
\[
\alpha_\Pi := \max_{F \in \F} w(F),
\]
which is equal to $w(Y)$ for all $Y \in \sol(I_\Pi)$ by definition. 
Given some $d' : \U'_L \to \R$, let $Y_{d'}$ be an optimal follower's response to $d'$, chosen optimistically, 
i.e., a set $Y \in \F$ that maximizes $p'(Y) - d'(Y)$, and among all these sets one that maximizes $d'(Y)$. 
Note that $Y_{d'}$ is not necessarily unique, but this will not influence the correctness of our arguments.

\begin{lemma}
\label{lem:follower-opt} 
  No matter which $d'$ the leader selects in \eqref{eq:pi-meta-thm}, the follower's optimal value is always at least $\alpha_\Pi M$.
\end{lemma}
\begin{proof}
  Let $d' : \U'_L \to \R$ be arbitrary.
Consider problem~\eqref{eq:SAT-meta-thm}. We can assume w.l.o.g. that \eqref{eq:SAT-meta-thm} is bounded, i.e., that there is some $X \in \sol(I_\text{\sat{}})$ with $X \subseteq \U_F$.
Then, by the SSP property, there is some $Y \in \sol(I_\Pi)$ such that $Y \cap f(\U_\text{\sat{}}) = f(X)$. 
This implies $Y \subseteq \U'_F$  (see \cref{fig:schematic-meta-reduction}). 
Therefore, $d'(Y) = 0$ and $p'(Y) - d'(Y) = p'(Y) \geq M w(Y) = \alpha_\Pi M$, since $Y \in \sol(I_\Pi)$.
This implies that, no matter which function $d'$ the leader selects, the follower's optimal value is always at least~$\alpha_\Pi M$.  \qed
\end{proof}

Next, consider the following definition. Let $d' \colon \U'_L \to \R$ be called \emph{reasonable} if $d'(Y_{d'}) \geq 0$ 
, i.e., the leader expects to gain any profit.

\begin{lemma} 
\label{lem:F-Sol-equivalence}
If we assume  that $d'(Y_{d'}) \geq 0$, then the following holds:
\begin{equation} \label{eq:F-Sol-Equivalence}
\begin{aligned}
&\max\set{d'(Y^\star) : Y^\star \in \argmax_{Y \in \F}\set{p'(Y) - d'(Y)}} \\
= &\max\set{d'(Y^\star) : Y^\star \in \argmax_{Y \in \sol(I_\Pi)}\set{p'(Y) - d'(Y)}}.
\end{aligned}
\end{equation}
\end{lemma}
 Roughly speaking, \cref{lem:F-Sol-equivalence} states that we can assume w.l.o.g. in \eqref{eq:pi-meta-thm} that the follower optimizes only over $\sol(I_\Pi)$ instead of over $\F \supseteq \sol(I_\Pi)$. 
 Intuitively, \cref{lem:F-Sol-equivalence} follows from the fact that $M$ is very large. 
 However, note that \cref{lem:F-Sol-equivalence} cannot hold without any assumptions on $d'$. 
 If the leader chooses a possibly very unreasonable $d' \ll 0$ with $\lvert d'(e) \rvert \gg \lvert p'(e) \rvert$ for some $e \in \U'_L$, the best solution for the follower might not be contained in $\sol(I_\Pi)$ anymore. 
\cref{lem:F-Sol-equivalence} then establishes that, for all \enquote{natural} $d'$, this does not happen.
 
\begin{proof} 
Let us assume that $d'(Y_{d'}) \geq 0$.
First, we claim that
\begin{equation}
\max_{Y \in \F}\set{p'(Y) - d'(Y)} = \max_{Y \in \sol(I_\Pi)}\set{p'(Y) - d'(Y)}.
\label{eq:max-are-equal}
\end{equation}
Note that, by definition,  $Y_{d'}$ is an optimal follower's response to $d'$, 
hence the left-hand side of  $\eqref{eq:max-are-equal}$ is equal to $p'(Y_{d'}) - d'(Y_{d'})$. 
We show that the right-hand side is equal to $p'(Y_{d'}) - d'(Y_{d'})$ as well. Since $\sol(I_\Pi) \subseteq \F$, it suffices to show that $Y_{d'} \in \sol(I_\Pi)$.
For the sake of contradiction, assume $Y_{d'} \not\in \sol(I_\Pi)$. Then $w(Y_{d'}) < \alpha_\Pi$. Hence,
\[
p'(Y_{d'}) - d'(Y_{d'}) \leq (\alpha_\Pi - 1)M + p(\U_\text{\sat{}}) < \alpha_\Pi M.
\]
But, by the definition of $Y_{d'}$, the term $p'(Y_{d'}) - d'(Y_{d'})$ is the follower's optimal value, so this is a contradiction to \cref{lem:follower-opt}. 

 Now, note that \eqref{eq:max-are-equal} implies 
\begin{equation}
\argmax_{Y \in \F}\set{p'(Y) - d'(Y)} \ \supseteq \ \ \argmax_{Y \in \sol(I_\Pi)}\set{p'(Y) - d'(Y)}.
\label{eq:argmax-equal}
\end{equation}
Furthermore, $Y_{d'}$ is by definition an element maximizing $d'(Y)$ over the set on the left-hand side in \eqref{eq:argmax-equal}.
Furthermore, we know that $Y_{d'} \in \sol(I_\Pi)$. 
Therefore, \eqref{eq:F-Sol-Equivalence} is true. \qed
\end{proof}

\begin{lemma}
\label{lem:meta-thm-if}
If $\exists A \forall B \varphi(A,B)$ is satisfiable, then \eqref{eq:pi-meta-thm} has value at least $k^\star$.
\end{lemma}
\begin{proof}
In this case, due to \cref{thm:pricing-SAT}, there exists some $d^\star : \U_L \to \R$ that achieves a value of at least $k^\star$ in \eqref{eq:SAT-meta-thm}. 
Furthermore, we can assume w.l.o.g. that $d^\star \geq 0$, since the function $d^\star$ explicitly constructed in \cref{thm:pricing-SAT} has this property.
Let us define $d'^\star : \U'_L \to \R$ via $d'^\star(e) := d^\star(f^{-1}(e))$ for all $e \in \U'_L$.
Then, in particular, we have $d'^\star \geq 0$, thus, $d'^\star(Y_{d'^\star}) \geq 0$, so \cref{lem:F-Sol-equivalence} is applicable.
Let us say that the leader chooses to use the price function $d'^\star$ in \eqref{eq:pi-meta-thm}.
Then, since we consider the optimistic setting, the leader receives a payout of 
\begin{align*}
&\max\set{d'^\star(Y^\star) : Y^\star \in \argmax_{Y \in \F}\set{p'(Y) - d'^\star(Y)}}\\
 = &\max\set{d'^\star(Y^\star) : Y^\star \in \argmax_{Y \in \sol(I_\Pi)}\set{p'(Y) - d'^\star(Y)}}\\
 = &\max\set{d'^\star(Y^\star) : Y^\star \in \argmax_{Y \in \sol(I_\Pi)}\set{p(f^{-1}(Y)) - d'^\star(Y)}}\\
  = &\max\set{d^\star(X^\star) : X^\star \in \argmax_{X \in \sol(I_\text{\sat{}})}\set{p(X) - d^\star(X)}}\\
  \geq &\ k^\star.
\end{align*}
Here, the first line is by definition of \eqref{eq:pi-meta-thm}, the second line is by \cref{lem:F-Sol-equivalence},
the third line holds due to the cost structure of $p'$, and since $\alpha_\Pi M$ is a constant. The fourth line holds since $d'^\star(Y) = d^\star(f^{-1}(Y))$ and since, by the SSP property, the following two sets are equal:
\[
\set{f^{-1}(Y) : Y \in \sol(I_\Pi)} = \sol(I_\text{\sat{}})
\]
In total, we conclude that the optimal value of \eqref{eq:pi-meta-thm} is at least $k^\star$. \qed
\end{proof}

\begin{lemma}
\label{lem:meta-thm-only-if}
If \eqref{eq:pi-meta-thm} has value at least $k^\star$, then $\exists A \forall B \varphi(A,B)$ is satisfiable.
\end{lemma}
\begin{proof}
Let $d'^\star : \U'_L \to \R$ be a function that achieves value at least $k^\star$ in \eqref{eq:pi-meta-thm}. 
Let us define $d^\star : \U_L \to \R$ via $d^\star(e) := d'^\star(f(e))$ for all $e \in \U_L$.
We know that $d'^\star(Y_{d'^\star}) \geq k^\star \geq 0$, so $d'^\star$ is reasonable and \cref{lem:F-Sol-equivalence} is applicable. 
Then we can perform the reasoning of \cref{lem:meta-thm-if} in reverse. 
Specifically, assume that the leader chooses a cost function $d^\star$ in \eqref{eq:SAT-meta-thm}. 
Then they receive a payout of
\begin{align*}
  &\max\set{d^\star(X^\star) : X^\star \in \argmax_{X \in \sol(I_\text{\sat{}})}\set{p(X) - d^\star(X)}}\\
   = &\max\set{d'^\star(Y^\star) : Y^\star \in \argmax_{Y \in \sol(I_\Pi)}\set{p(f^{-1}(Y)) - d^\star(f^{-1}(Y))}}\\
    = &\max\set{d'^\star(Y^\star) : Y^\star \in \argmax_{Y \in \sol(I_\Pi)}\set{p'(Y) - d'^\star(Y)}}\\
    = &\max\set{d'^\star(Y^\star) : Y^\star \in \argmax_{Y \in \F}\set{p'(Y) - d'^\star(Y)}}\\
    \geq & \ k^\star.
\end{align*}
This implies that the optimal value of \eqref{eq:SAT-meta-thm} is at least $k^\star$. Therefore, $\exists A \forall B \varphi(A,B)$ is satisfiable. \qed
\end{proof}

In total, \cref{lem:meta-thm-if,lem:meta-thm-only-if} prove our main theorem, \cref{thm:meta-reduction-max}.

\begin{corollary}
\label{cor:d-well-behaved-max}
\cref{thm:meta-reduction-max} remains true, even if any of the additional constraints $0 \leq d$, $d \leq p$, or $0 \leq d \leq p$ 
is added in \eqref{eq:pricing-pi-max-2}.
\end{corollary}
\begin{proof}
We are now interested in the pricing problem \eqref{eq:pi-meta-thm} with the additional constraint $0 \leq d' \leq p'$.
Consider the reduction from $\exists \forall$-DNF-\sat{} to \textsc{Pricing-$\Pi$} performed in \cref{thm:meta-reduction-max}.
We check the reasoning of \cref{lem:meta-thm-if}.
If we assume that $\exists A \forall B \varphi(A,B)$ is satisfiable, then $d^\star$ satisfies $0 \leq d^\star \leq p$. 
Then, by the definition of $p'$ and by the definition of $d'^\star$ in \cref{lem:meta-thm-if}, we also have $0 \leq d'^\star \leq p'$.
On the other hand, if we start \cref{lem:meta-thm-only-if} with the stronger assumption of $0 \leq d'^\star \leq p'$, of course the reasoning of the lemma remains true.
In total, we have shown that the problem \textsc{Pricing-$\Pi$} remains hard, even if the constraint $0 \leq d' \leq p'$ is added to \eqref{eq:pi-meta-thm}.
The proof for the other choices of the additional constraint follows analogously.
\qed
\end{proof}

\begin{theorem}
\label{thm:meta-reduction-feas}
  For any problem $\Pi$ that is $\SSPNP$-complete (see \cref{sec:framework}), the problem \textsc{Feas-Pricing-$\Pi$} as in \eqref{eq:pricing-pi-feas} is $\Sigma^p_2$-hard.
\end{theorem}
\begin{proof}
We can interpret the SSP problem $\Pi = (\I, \U, \sol)$ trivially as a maximization LOP problem by defining, for each instance $I \in \I$, the sets $\F(I) := \sol(I)$ 
with weights $w^{(I)} : \U(I) \to \Z_{\geq 0}$, $w^{(I)}(e) = 0$ for all $e \in \U(I)$, and threshold $t^{(I)} = 0$. 
Note that then $\F(I) = \sol(I) = \set{F \in \F(I) : w^{(I)}(F) = 0}$. In this interpretation, the problem $\Pi$ is $\LOP$-complete. Furthermore, under this interpretation, the problem \eqref{eq:pricing-pi-max-2} becomes equal to \eqref{eq:pricing-pi-feas}.
Then the theorem follows as a consequence of \cref{thm:meta-reduction-max}.
\qed
\end{proof}
\section{Minimization Case}

We now describe the necessary modifications to adapt the main theorem, \cref{thm:meta-reduction-max}, to the minimization case. 
We again start the reduction from an instance $\exists A \forall B \varphi(A,B)$ of $\exists \forall$DNF-\sat{}. 
Using \cref{thm:pricing-SAT}, we compute in polynomial time
an instance of \textsc{Feas-Pricing-\sat{}} and a number $k^\star$ such that \eqref{eq:SAT-meta-thm-min} has value at least $k^\star$ if and only if $\exists A \forall B \varphi(A,B)$ is satisfiable.

\begin{equation} \label{eq:SAT-meta-thm-min}
  \begin{aligned}
    \max_{d \colon \U_L \to \R}\ & d(X^\star)\\
    \text{s.t. } & X^\star \in \argmax_{X \in \sol(I_\text{\sat{}})}\ p(X) - d(X). \\
  \end{aligned}
\end{equation}

Since $\Pi$ is $\LOP$-complete and a minimization problem, one can obtain in polynomial time an instance $I_\Pi$ of $\Pi$ together with a universe $\U_\Pi$, feasible solutions $\F$, weights
$w : \U_\Pi \to \Z_{\geq 0}$, a threshold $t \in \Z_{\geq 0}$, and an injective function $f : \U_\text{\sat{}} \to \U_\Pi$ such that the following SSP property holds:
\begin{align*}
\sol(I_\Pi) := \set{F \in \F &: w(F) \leq t} = \set{F \in \F : w(F) = t}\\
\text{ such that} \quad &\forall S' \in \sol(I_\Pi) : f^{-1}(S') \in \sol(I_\text{\sat{}})\\
\text{ and} \quad &\forall S \in \sol(I_\text{\sat{}}): \exists S' \in \sol(I_\Pi) : S' \cap f(\U_\text{\sat{}}) = f(S).
\end{align*}

For technical reasons, we would like to assume the following:
\begin{lemma}
\label{lem:w-at-least-one}
Without loss of generality, we have $w(e) \geq 1$ for all $e \in f(\U_\text{\sat{}})$. 
\end{lemma}
\begin{proof}
Suppose that this is not the case. Then let us define an alternative weight function $w' : \U_\Pi \to \Z_{\geq 0}$ as follows.
Set $n :=  |\U_\text{\sat{}}|$ and $K := n+1$. Then let $w'(e) := Kw(e) + 1$ for all $e \in f(\U_\text{\sat{}})$ and $w'(e) := Kw(e)$ for all other $e \in \U_\Pi \setminus f(\U_\text{\sat{}})$.
Furthermore, define a new threshold $t' := Kt + n/2$ (note that $n$ is even since $\U_\text{\sat{}}$ consists of positive and negative literals).
Clearly, $w'(e) \geq 1$ holds for all $e \in f(\U_\text{\sat{}})$. 
Now, for all $F \in \F$, we have that $w'(F) \leq t'$ implies $F \in \sol(I_\Pi)$ by the choice of $K$. 
Furthermore, for all $F \in \sol(I_\Pi)$, we have $|F \cap f(U_\text{\sat{}})| = n/2$ due to the properties of the SSP reduction, and $|S| = n/2$ for all $S \in \sol(I_\text{\sat{}})$. 
Then we have the equivalence $w(F) \leq t \Leftrightarrow w(F) = t \Leftrightarrow w'(F) \leq t' \Leftrightarrow w'(F) = t'$. Hence, $w'$ and $t'$ still have the properties required for an SSP reduction. \qed
\end{proof}

Just like in \cref{sec:meta-reduction}, we now consider the problem $\textsc{Pricing-$\Pi$}$ with the underlying instance $I_\Pi$, 
the leader's universe $\U'_L := f(\U_L)$, and the follower's universe $\U'_F := f(\U_F) \cup (\U_\Pi \setminus f(\U_\text{\sat{}}))$. 
Let again $n := |\U_\text{\sat{}}|$ and $M := 4n \sum_{e \in \U_\text{\sat{}}} c(e)$. Note that $M \geq 0$.
The cost function $c' : \U'_L \cup \U'_F \to \Zgeq$ is defined as before, but with the difference that we subtract $p(f^{-1}(e))$ instead of adding it:
\[
c'(e) := 
\begin{cases}
Mw(e) - p(f^{-1}(e)) &\text{ if }e \in f(\U_\text{\sat{}})\\
Mw(e)				&\text{ if }e \in \U_\Pi \setminus f(\U_\text{\sat{}})\\
\end{cases}
\]
This gives rise to the following instance of $\textsc{Pricing-$\Pi$}$:
\begin{equation} \label{eq:pi-meta-thm-min}
  \begin{aligned}
    \max_{d' \colon \U'_L \to \R}\ & d'(Y^\star)\\
    \text{s.t. } & Y^\star \in \argmin_{Y \in \F}\ c'(Y) + d'(Y) \\
  \end{aligned}
\end{equation}

Note that $c' \geq 0$, due to our assumption in \cref{lem:w-at-least-one}. Therefore, this is a well-defined instance of $\textsc{Pricing-$\Pi$}$. 
In the remainder of this section, we show that $\exists A \forall B \varphi(A, B)$ is true if and only if both \eqref{eq:SAT-meta-thm-min} and \eqref{eq:pi-meta-thm-min} have value at least~$k^\star$.
Notice the following subtle difference: The follower in \eqref{eq:SAT-meta-thm-min} solves a maximization problem, while the follower in \eqref{eq:pi-meta-thm-min} solves a minimization problem.
Again, let $\alpha_\Pi$ denote the optimal value of the minimization instance $I_\Pi$, i.e.,
\[
\alpha_\Pi := \min_{F \in \F} w(F),
\]
which is equal to $w(Y)$ for all $Y \in \sol(I_\Pi)$. Like in \cref{sec:meta-reduction}, for all $d' : \U'_L \to \R$, we define the set $Y_{d'} \in \F$ to be a best follower's response to $d'$, 
chosen optimistically, i.e., a set that minimizes $c'(Y) + d'(Y)$ over $\F$, and among all these sets one that maximizes $d'(Y)$.
We then obtain the following lemmas, which are analogous to the ones in \cref{sec:meta-reduction}.

\begin{lemma}
\label{lem:follower-opt-min} 
No matter which $d'$ the leader selects in \eqref{eq:pi-meta-thm-min}, the follower's optimum is always at most $\alpha_\Pi M$.
\end{lemma}
\begin{proof}
We argue analogously to \cref{lem:follower-opt}. W.l.o.g., there is a set $S \in \sol(I_\text{\sat{}})$ with $S \subseteq \U_F$.
Hence, there is some $Y \in \sol(I_\Pi)$ with $Y \subseteq \U'_F$ and $c'(Y) + d'(Y) = c'(Y) \leq \alpha_\Pi M$.
\qed
\end{proof}

\begin{lemma} 
\label{lem:F-Sol-equivalence-min}
If we assume that $d'(Y_{d'}) \geq 0$, then the following holds:
\begin{equation} \label{eq:F-Sol-Equivalence-min}
\begin{aligned}
&\max\set{d'(Y^\star) : Y^\star \in \argmin_{Y \in \F}\set{c'(Y) + d'(Y)}} \\
= &\max\set{d'(Y^\star) : Y^\star \in \argmin_{Y \in \sol(I_\Pi)}\set{c'(Y) + d'(Y)}}.
\end{aligned}
\end{equation}
\end{lemma}
\begin{proof}
The proof is very similar to the proof of \cref{lem:F-Sol-equivalence}. 
We first claim that $Y_{d'} \in \sol(I_\Pi)$. Indeed, assume the opposite, then $w(Y_{d'}) > \alpha_\Pi$, and
\[
c'(Y_{d'}) + d'(Y_{d'}) > (\alpha_\Pi + 1)M - c(\U_\text{\sat{}}) + 0 > \alpha_\Pi M.
\]
This is a contradiction to \cref{lem:follower-opt-min}, which states that $c'(Y_{d'}) + d'(Y_{d'}) \leq \alpha_\Pi M$.
The rest of the proof is analogous. \qed
\end{proof}

\begin{lemma}
\label{lem:meta-thm-if-min}
If $\exists A \forall B \varphi(A,B)$ is satisfiable, then \eqref{eq:pi-meta-thm-min} has value at least $k^\star$.
\end{lemma}
\begin{proof}
Since $\exists A \forall B \varphi(A,B)$ is satisfiable, there exists some $d^\star : \U_L \to \R$ that achieves a value of at least $k^\star$ in \eqref{eq:SAT-meta-thm-min}. 
Let us define $d'^\star : \U'_L \to \R$ via $d'^\star(e) := d^\star(f^{-1}(e))$ for all $e \in \U'_L$.
Again, we can assume w.l.o.g. that $d^\star \ge 0$ by \cref{thm:pricing-SAT}.
Then, in particular, we have $d'^\star \ge 0$, thus, $d'^\star(Y_{d'^\star}) \geq 0$, so \cref{lem:F-Sol-equivalence-min} is applicable.
Let us say that the leader chooses to use the price function $d'^\star$ in \eqref{eq:pi-meta-thm-min}. Then, since we consider the optimistic setting, the leader a receives payout~of
\begin{align*}
&\max\set{d'^\star(Y^\star) : Y^\star \in \argmin_{Y \in \F}\set{c'(Y) + d'^\star(Y)}}\\
 = &\max\set{d'^\star(Y^\star) : Y^\star \in \argmin_{Y \in \sol(I_\Pi)}\set{c'(Y) + d'^\star(Y)}}\\
 = &\max\set{d'^\star(Y^\star) : Y^\star \in \argmin_{Y \in \sol(I_\Pi)}\set{-p(f^{-1}(Y)) + d'^\star(Y)}}\\
 = &\max\set{d'^\star(Y^\star) : Y^\star \in \argmax_{Y \in \sol(I_\Pi)}\set{p(f^{-1}(Y)) - d^\star(f^{-1}(Y))}}\\
  = &\max\set{d^\star(X^\star) : X^\star \in \argmax_{X \in \sol(I_\text{\sat{}})}\set{p(X) - d^\star(X)}}\\
  \geq &\ k^\star.
\end{align*}
Here, the first line is by definition of \eqref{eq:pi-meta-thm-min}, the second line is by \cref{lem:F-Sol-equivalence-min}, 
the third line holds due to the cost structure of $c'$, and since $\alpha_\Pi M$ is a constant. The fourth line holds since $d'^\star(Y) = d^\star(f^{-1}(Y))$, and since $\argmin_x g(x) = \argmax_x -g(x)$.
The fifth line follows since, by the SSP property, the two following two sets are equal:
\[
\set{f^{-1}(Y) : Y \in \sol(I_\Pi)} = \sol(I_\text{\sat{}}).
\]
In total, we conclude that the optimal value of \eqref{eq:pi-meta-thm-min} is at least $k^\star$. \qed
\end{proof}

\begin{lemma}
\label{lem:meta-thm-only-if-min}
If \eqref{eq:pi-meta-thm-min} has value at least $k^\star$, then $\exists A \forall B \varphi(A,B)$ is satisfiable.
\end{lemma}
\begin{proof}
Let $d'^\star : \U'_L \to \R$ be a function that achieves value at least $k^\star$ in \eqref{eq:pi-meta-thm-min}. 
Let us define $d^\star : \U_L \to \R$ via $d^\star(e) := d'^\star(f(e))$ for all $e \in \U_L$.
We know that $d'^\star(Y_{d'^\star}) \geq k^\star \geq 0$, and so \cref{lem:F-Sol-equivalence-min} is applicable. 
Then we can perform the reasoning of \cref{lem:meta-thm-if-min} in reverse. 
Specifically, assume that the leader chooses cost function $d^\star$ in \eqref{eq:SAT-meta-thm-min}. 
Then they receive a payout of
\begin{align*}
  &\max\set{d^\star(X^\star) : X^\star \in \argmax_{X \in \sol(I_\text{\sat{}})}\set{p(X) - d^\star(X)}}\\
   = &\max\set{d'^\star(Y^\star) : Y^\star \in \argmax_{Y \in \sol(I_\Pi)}\set{p(f^{-1}(Y)) - d^\star(f^{-1}(Y))}}\\
    = &\max\set{d'^\star(Y^\star) : Y^\star \in \argmin_{Y \in \sol(I_\Pi)}\set{-p(f^{-1}(Y)) + d'^\star(Y)}}\\
    = &\max\set{d'^\star(Y^\star) : Y^\star \in \argmin_{Y \in \sol(I_\Pi)}\set{c'(Y) + d'^\star(Y)}}\\
    = &\max\set{d'^\star(Y^\star) : Y^\star \in \argmin_{Y \in \F}\set{c'(Y) + d'^\star(Y)}}\\
    \geq & \ k^\star.
\end{align*}
This implies that the optimal value of \eqref{eq:SAT-meta-thm-min} is at least $k^\star$. Therefore, $\exists A \forall B \varphi(A,B)$ is satisfiable. \qed
\end{proof}

In conclusion, we obtain:
\begin{theorem}
\label{thm:meta-reduction-min}
For any problem $\Pi$ that is an $\LOP$-complete minimization LOP problem (see \cref{sec:framework}), the problem \textsc{Pricing-$\Pi$} as in \eqref{eq:pricing-pi-min} is $\Sigma^p_2$-hard.
\end{theorem}

We can strengthen the result as follows.
\begin{corollary}
\label{cor:d-well-behaved-min}
\cref{thm:meta-reduction-min} remains true, even if any of the additional constraints $0 \leq d$, $d \leq c$, $0 \leq d \leq c$, or $-c \leq d$
is added in \eqref{eq:pricing-pi-min}.
\end{corollary}
\begin{proof}
The proof is analogous to the one of \cref{cor:d-well-behaved-max}.
First, in \cref{lem:meta-thm-if-min}, if $\exists A \forall B \varphi(A,B)$ is satisfiable, we can assume $0 \leq d^\star \leq p$ by \cref{thm:pricing-SAT}, and this implies $0 \leq d'^\star \leq c'$ since $w(e) \geq 1$ for all $e \in \U'_L$.
On the other hand, if we start \cref{lem:meta-thm-only-if-min} with the stronger assumption of $0 \leq d'^\star \leq c'$, of course the reasoning of the lemma remains true. 
The same is true for the other constraints listed above. Note that, in the minimization setting, also the constraint $-c \leq d$ can be argued to be natural (since the follower's costs never become negative in this case).
\qed
\end{proof}

\section{Conclusion}
\label{sec:conclusion}
In this work, we showed that the knapsack pricing problem is $\Sigma^p_2$-complete, answering an open question of Pferschy, Nicosia, Pacifi, and Schauer \cite{pferschy2021stackelberg} and a conjecture of Böhnlein, Schaudt, and Schauer \cite{bohnlein2023stackelberg}.
Furthermore, we extended this result by showing that, when considering an arbitrary $\NP$-complete problem, only very few assumptions are needed to ensure that the corresponding pricing problem is $\Sigma^p_2$-complete. 
These natural assumptions are met by most classical $\NP$-complete problems.
Our result holds for both natural maximization and minimization variants of the pricing problem, 
and for any choice of constraints $0 \leq d$, $d\leq c$, or $0 \leq d \leq c$ for the leader.

The approach of this paper is limited to underlying problems that are $\NP$-hard. 
It would be very interesting to extend our techniques to underlying problems in $\PTIME$, showing $\NP$-completeness of the pricing problem in a unified way.
However, we do not know if such a result is possible.

\bibliographystyle{splncs04}
\bibliography{literature.bib}

\end{document}